\documentclass[runningheads]{llncs}

\usepackage{amsmath}
\usepackage{amssymb}
\usepackage[normalem]{ulem}
\usepackage{color}
\usepackage{float}
\usepackage{graphicx}
\usepackage[loose]{subfigure}
\usepackage{cite}

\begin{document}
\newcommand{\ABV}{\ensuremath{AB1V} }
\newcommand{\ABVp}{\ensuremath{AB1V}}
\newcommand{\NP}{\ensuremath{\mathcal{NP}}}
\newcommand{\E}{\ensuremath{\mathcal{E}}}
\def\O(#1){\ensuremath{\mathcal{O}(#1)}}

\title{On Fan-Crossing Graphs}
\titlerunning{On Fan-Crossing Graphs}

\author{Franz J.\ Brandenburg}
\institute{94030 Passau, Germany \\
\email{brandenb@informatik.uni-passau.de}}

\authorrunning{F.\ J.\ Brandenburg}

\maketitle

\begin{abstract}
A \emph{fan} is a set of edges with a single common endpoint. A
graph is \emph{fan-crossing} if it admits a drawing in the plane so
that each edge is crossed by edges of a fan.
 It is \emph{fan-planar}
if, in addition, the common endpoint is on the same side of the
crossed edge. A graph is \emph{adjacency-crossing} if it admits a
drawing so that crossing edges are adjacent. Then it excludes
\emph{independent crossings} which are crossings by edges with no
common endpoint. Adjacency-crossing allows \emph{triangle-crossings}
in which an edge crosses the edges of a triangle, which is excluded
at fan-crossing graphs.

We  show that every  adjacency-crossing graph is fan-crossing. Thus
triangle-crossings can be avoided. On the other hand, there are
fan-crossing graphs that are not fan-planar, whereas for every
fan-crossing graph there is a fan-planar graph on the same set of
vertices and with the same number of edges. Hence, fan-crossing and
fan-planar graphs are different, but they do not differ in their
density with at most $5n-10$ edges for graphs of size $n$.
\end{abstract}



\section{Introduction} \label{sect:intro}

Graphs with  or without special patterns  for edge crossings  are an
important topic in Topological Graph Theory, Graph Drawing, and
Computational Geometry. Particular patterns
  are no crossings, single crossings, fans, independent edges, or no three
  pairwise crossing edges.
  A \emph{fan} is a set of edges with a single common endpoint.
  In complement, edges are \emph{independent} if they do not share a common endpoint.
Important graph classes have been defined in this way, including
 the planar, 1-planar
\cite{klm-bib-17, ringel-65}, fan-planar \cite{bddmpst-fan-15,
bcghk-rfpg-17,ku-dfang-14}, fan-crossing free
  \cite{cpkk-fan-15}, and quasi-planar  graphs \cite{aapps-qpg-97}.
  A first order logic definition of these and other graph classes is given in \cite{b-FOL-17}.
These definitions are motivated by the need for classes of
non-planar graphs from real world applications, and a negative
correlation between   edge crossings   and the readability of graph
drawings by human users. The aforementioned graph classes aim to
meet both requirements.

We consider undirected graphs $G = (V,E)$ with   finite sets of
vertices $V$ and edges $E$ that are \emph{simple} both in a graph
theoretic and in a topological sense. Thus we do not admit multiple
edges and self-loops, and we exclude multiple crossings of two edges
and crossings among adjacent edges.

A \emph{drawing}   of a graph $G$ is a mapping of $G$ into the plane
so that the vertices are mapped to distinct points and each edge is
mapped to a Jordan arc between the endpoints. Two edges \emph{cross}
if their Jordan arcs intersect in a point other than an endpoint.
Crossings subdivide an edge into uncrossed pieces, called \emph{edge
segments}, whose endpoints are vertices or crossing points. An edge
is \emph{uncrossed} if and only if it consists of a single edge
segment.
A drawn graph is called a \emph{topological graph}. In other works,
a topological graph is called an \emph{embedding}  which is the
class of topologically equivalent drawings.
An embedding defines a \emph{rotation system} which is the cyclic
sequence of edges incident to each vertex. A drawn graph partitions
the plane into topologically connected regions, called \emph{faces}.
The unbounded region is called the \emph{outer face}. The
\emph{boundary} of each face consists of  a cyclic sequence of edge
segments. It is commonly specified by the sequence of vertices and
crossing points of the edge segments.
%
 The subgraph
of a graph $G$ induced by a subset $U$ of vertices is denoted by
$G[U]$. It inherits its embedding from an embedding of $G$, from
which all vertices not in $U$ and all edges with at most one
endpoint in $U$ are removed.

\begin{figure}
  \centering
  \subfigure[ ]{
    \parbox[b]{2.8cm}{%
      \centering
      \includegraphics[scale=0.3]{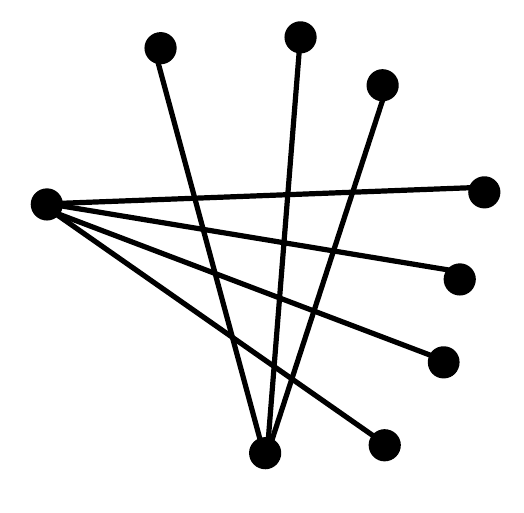}
    }
    \label{fig:simplefan}
  }
  \hfil
  \subfigure[ ]{
    \parbox[b]{2.8cm}{%
      \centering
      \includegraphics[scale=0.3]{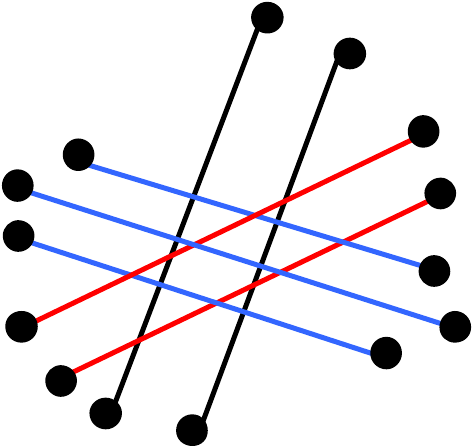}
    }
    \label{fig:fcf}
  }
  \caption{(a) A fan-crossing and
  (b)  an independent crossing or fan-crossing free }
  \label{fig:typeedgecrossings}
\end{figure}

An edge $e$ has a \emph{fan-crossing} if the crossing edges form a
fan, as in Fig.~\ref{fig:simplefan}, and an \emph{independent
crossing} if the crossing edges are independent, see
Fig.~\ref{fig:fcf}. Fan-crossings are also known as radial $(k,1)$
grid crossings   and independent crossings as grid crossings
 \cite{afps-grids-14}.
Independent crossings are excluded if and
only if \emph{adjacency-crossings} are allowed in which two edges
are adjacent if they both cross an edge \cite{b-FOL-17}.

 \emph{Fan-planar} graphs were introduced by Kaufmann and
Ueckerdt \cite{ku-dfang-14}, who imposed a special restriction,
called \emph{configuration II}. It is shown in Fig.~\ref{fig:conf2}.
Let $e,f$ and $g$ be three edges in a drawing so that $e$ is crossed
by $f$ and $g$, and $f$ and $g$ share a common vertex $t$. Then they
form   configuration II  if one endpoint of $e$ is inside a cycle
through $t$ with segments of $e, f$ and $g$, and the other endpoint
of $e$ is outside this cycle.  If $e = \{u,v\}$ is oriented from $u$
(left) to $v$ (right) and $f$ and $g$ are oriented away from $t$,
then $f$ and $g$ cross $e$ from different directions. Configuration
II admits \emph{triangle-crossings} in which an edge crosses the
edges of a triangle, see Fig.~\ref{fig:conf2triangle}. Observe that
a triangle-crossing is the only configuration in which an edge is
crossed by edges that do not form a fan and that are not
independent.

\begin{figure}
  \centering
  \subfigure[ ]{
    \parbox[b]{4.5cm}{%
      \centering
      \includegraphics[scale=0.6]{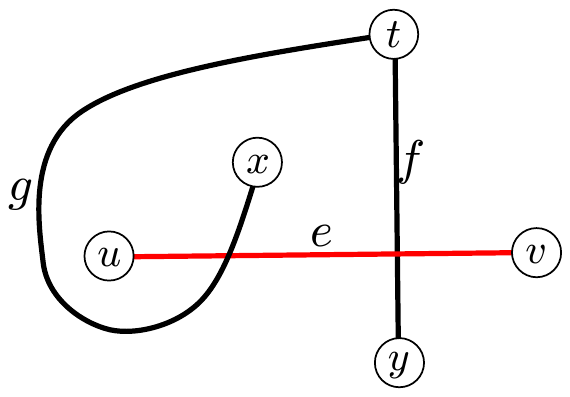}
    }
    \label{fig:conf2}
  }
  \hfil
  \subfigure[ ]{
    \parbox[b]{4.5cm}{%
      \centering
      \includegraphics[scale=0.6]{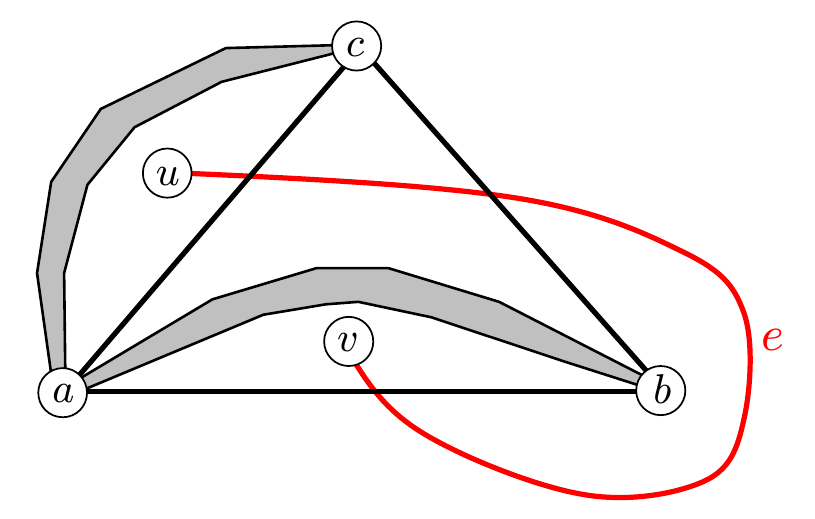}
    }
    \label{fig:conf2triangle}
  }
 \caption{(a) Configuration II in which edge $e=\{u,v\}$ is crossed by edges $\{t,x\}$
    and $\{t,y\}$ and $x$ and $y$ are on opposite sides of $e$ and (b)  edge $e= \{u,v\}$
    crosses
    a triangle. The shaded regions represent subgraphs which shall prohibit another
    routing of $e$. Similar regions could be added to (a), as in Fig.~\ref{fig:graphM}.}
  \label{fig:fanconfiguration}
\end{figure}

A graph   is \emph{fan-crossing free} if it admits a drawing without
fan-crossings \cite{cpkk-fan-15}. Then there are only independent
crossings. A graph is \emph{fan-crossing} if it admits a drawing in
which each crossing is a
  fan-crossing, and  \emph{adjacency-crossing} if it can be
drawn so that each edge is crossed by edges that are adjacent. Then
 independent crossings are excluded. As stated in \cite{b-FOL-17}, adjacency
crossing is complementary to independent crossing, but the graph
classes are not complementary and both properly include the 1-planar
graphs. A graph is \emph{fan-planar} if it avoids independent
crossings and configuration II \cite{ku-dfang-14}.

Observe the subtle differences between adjacency-crossing,
fan-crossing, and fan-planar graphs, which each exclude independent
crossings, and in addition exclude triangle-crossings and
configuration II, respectively.
 Kaufmann and Ueckerdt \cite{ku-dfang-14} observed that   configuration II
cannot occur in straight-line drawings, so that every straight-line
adjacency-crossing drawing is fan-planar.
They proved that fan-planar graphs of size $n$ have at most $5n-10$
edges and posed the density of adjacency-crossing graphs as an open
problem. The \emph{density} defines an upper bound an on the number
of edges in graphs of size $n$.
 We show that triangle-crossings can be avoided by an edge
rerouting, and that configuration II can be  restricted to a special
case.
 Moreover, the allowance
or exclusion of configuration II   has no impact on the density,
which answers the  above question. In particular, we   prove the
following:

\begin{enumerate}
  \item Every adjacency-crossing graph is fan-crossing. Thus
triangle-crossings can be avoided.
  \item There are fan-crossings graphs that are not fan-planar.
Thus configuration II is essential.

  \item For every    fan-crossing  graph $G$
there is a fan-planar graph $G'$ on the same set of vertices and
with (at least) the same number of edges. Thus fan-crossing graphs
of size $n$ have at most $5n-10$ edges.
\end{enumerate}

We prove that triangle-crossings can be avoided by an  edge
rerouting in   Section \ref{sect:trianglecrossings}  study
configuration II in Section \ref{sect:fanplanar}. We conclude in
Section \ref{sect:conclusion} with some open problems on
fan-crossing graphs.

\section{Triangle-Crossings} \label{sect:trianglecrossings}

In this section, all embeddings $\mathcal{E}(G)$   are
adjacency-crossing or equivalently they exclude independent
crossings. We consider triangle-crossings and show that they can be
avoided by an edge rerouting. A \emph{rerouted edge} is denoted by
$\tilde{e}$ if $e$ is the original one. More formally, we transform
an adjacency-crossing embedding $\mathcal{E}(G)$ into an
adjacency-crossing embedding $\tilde{\mathcal{E}}(G)$ which differs
from $\mathcal{E}(G)$ in the embedding of the rerouted edges such
that $\tilde{e}$ does not cross a particular triangle   if $e$
crosses that triangle.

For convenience,  we assume that triangle-crossings are in a
\emph{standard configuration},
 in which a triangle $\Delta = (a,b,c)$ is crossed by
edges $e_1,\ldots, e_k$ for some $k \geq 1$ that cross each edge of
$\Delta$. We call each $e_i$ a \emph{triangle-crossing edge} of
$\Delta$. These edges are incident to a common vertex $u$ if $k \geq
2$. We assume that a triangle-crossing edge $e=\{u,v\}$ crosses
$\{a,c\}, \{b,c\}$ and $ \{a,b\}$ in this order and that $u$ is
outside $\Delta$. Then $v$ must be inside $\Delta$. All other cases
are similar exchanging inside and outside and the order in which the
edges of $\Delta$ are crossed.

We need some further notation. Let $fan(v)$ denote a subset of edges
incident to vertex $v$ that cross a particular edge. This is a
generic definition. If the crossed edge is given, then $fan(v)$ can
be retrieved from the embedding $\mathcal{E}(G)$. In general,
$fan(v)$ does not contain all edges incident to $v$. A \emph{sector}
is a subsequence of edges of $fan(v)$ properly between two edges
$\{v, s\}$ and $\{v,t\}$ in clockwise order. An edge  $e$ is
\emph{covered} by  a vertex $v$ if $e$ is crossed by at least two
edges incident to $v$ so that $fan(v)$ has at least two elements.
Let $cover(v)$ denote the set of edges covered by $v$. Note that
uncrossed edges and edges that are crossed only once are not
covered. If an edge $e$ is crossed by an edge $g= \{u,v\}$, then $e$
is a candidate for $cover(u)$ or $cover(v)$ and $e \not\in cover(w)$
for any other vertex $w \neq u,v$ except if $e$ crosses a triangle.
In fact, an edge $e=\{u,v\}$ is triangle-crossing   if and only if
$\{e\} = cover(x) \cap cover(y)$ for vertices $x \neq y$. To see
this, observe that $e \in cover(x)$ for $x = a,b,c$ if $e$ crosses a
triangle $\Delta = (a,b,c)$. Conversely, if $e$ is crossed by edges
$\{a, w_1\}, \{a, w_2\}$ and $\{b, w_3\}$ with $a \neq b$ and $w_1
\neq w_2$, then $w_1 =   w_3$ and $w_2=b$ (up to renaming) if there
are no independent crossings.

Triangle crossings are special. If  an edge $e$ crosses a triangle
$\Delta$, then $e$ cannot be crossed by any   edge other than the
edges of $\Delta$. In particular, $e$ cannot cross another triangle
or another triangle-crossing edge. But an  edge may be part of two
triangle-crossings,  as a common edge of two crossed triangles, as
shown in Fig.~\ref{fig:twotriangles}, or as   a triangle-crossing
edge of one triangle and an edge of another triangle, as shown in
Fig.~\ref{fig:tri2}, and both configurations can be combined.

\begin{figure}[t]
  \centering
  \subfigure[ ]{
    \parbox[b]{4cm}{%
      \centering
      \includegraphics[scale=0.6]{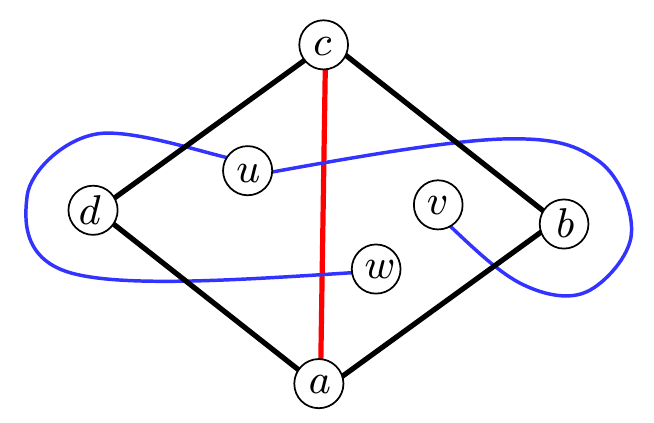}
    }
    \label{fig:twotriangles}
  }
  \hfil
  \subfigure[ ]{
    \parbox[b]{4cm}{%
      \centering
      \includegraphics[scale=0.6]{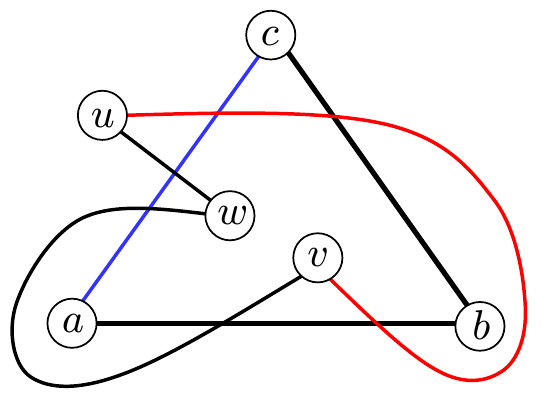}
    }
    \label{fig:tri2}
  }
 \caption{Two crossed triangles sharing  (a)  an edge
 or (b) an edge and a triangle-crossing edge.}
  \label{fig:triangleconfigurations}
\end{figure}

A particular example is $K_5$, which has five embeddings
\cite{hm-dcgmnc-92}, see Fig.~\ref{fig:allK5}. The one of
Fig.~\ref{fig:allK5}(e)   has a triangle-crossing. If it is a part
of an adjacency-crossing embedding, then we show that it can be
transformed into the embedding of Fig.~\ref{fig:allK5}(c) by
rerouting an edge of the crossed triangle.

\begin{figure}[h]
   \centering
   \subfigure[ ]{
     \includegraphics[scale=0.35]{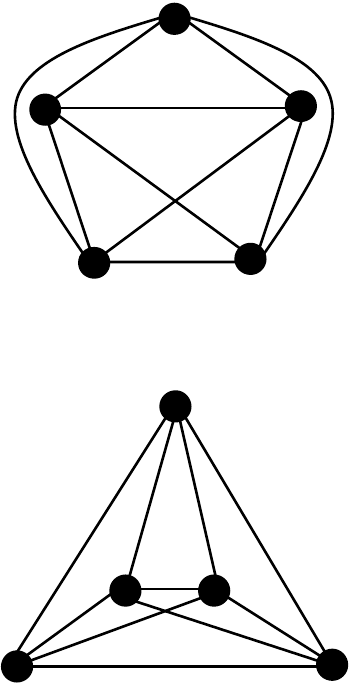}
   }
   \quad
      \subfigure[ ]{
      \rotatebox{1}{%
     \includegraphics[scale=0.35]{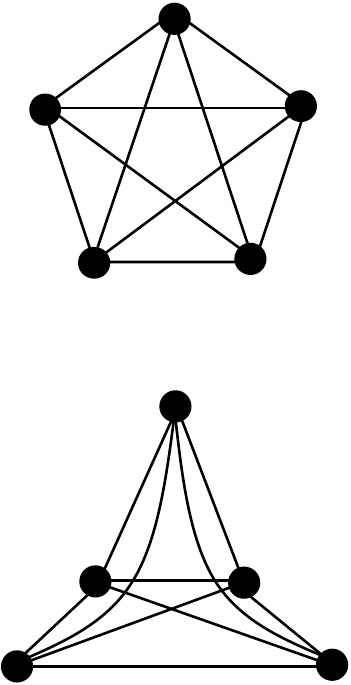}
   }
   }
   \quad
      \subfigure[ ]{
     \includegraphics[scale=0.35]{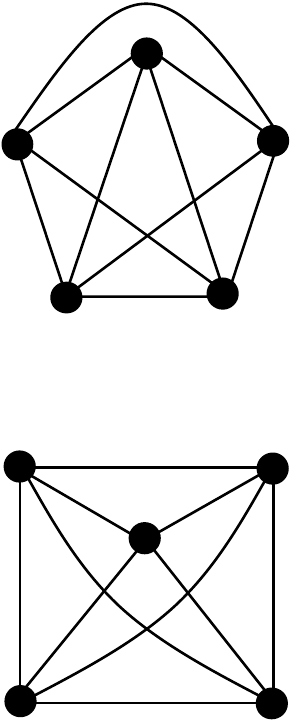}
   }
   \quad
      \subfigure[ ]{
     \includegraphics[scale=0.35]{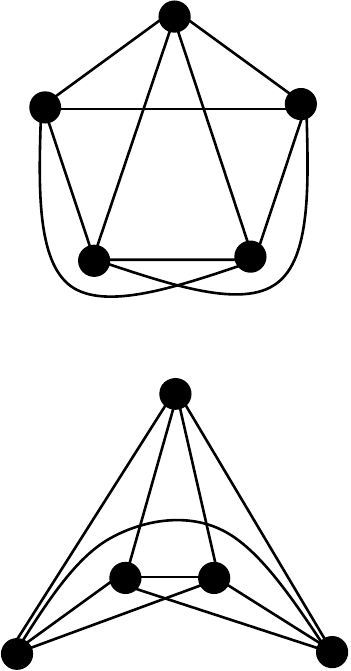}
   }
\quad
   \subfigure[ ]{
     \includegraphics[scale=0.35]{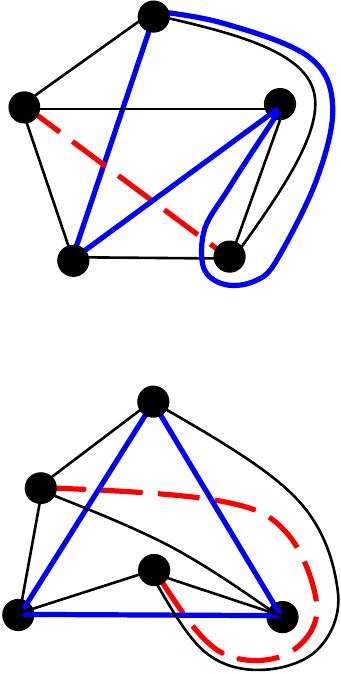}
   }
   \caption{All non-isomorphic embeddings of $K_5$ \cite{hm-dcgmnc-92} with two drawings. Only (a) is
   1-planar and fan-crossing free, (b), (c), and (d) are fan-planar
   and
 (e)  is  adjacency-crossing and has a triangle crossing with the
 triangle-crossing edge drawn red.
 Our rerouting transforms (e) into (c) and reroutes and straightens the curved edge.}
   \label{fig:allK5}
\end{figure}

In return, the edges of $\Delta$ can only be crossed by edges of
$fan(u)$ or $fan(v)$ if $e= \{u,v\}$ is a triangle-crossing edge of
$\Delta$. They are covered by $u$ if there are at least two
triangle-crossing edges incident to $u$.
 In addition, there may be edges that   cross only one or
two edges of   $\Delta$. These are incident to $u$ or $v$  and they
are incident to $u$ if there are at least two triangle-crossing
edges  incident to $u$. We assume a standard configuration and
classify crossing   edges   by the sequence of crossed edges, as
stated in Table~$1$.

\begin{table}[h]
\begin{tabular}{|l | c| l |}
  \hline
  name & set & sequence of crossed edges \\
  \hline
  needle & $N_1, N_2, N_3$ & $\{a,c\}$ \\
  $a$-hook & $H_a$&   $\{a,b\}$ \\
  $c$-hook & $H_c$  & $\{b,c\}$ \\
  $a$-arrow  & $A_a$ & $\{a,c\}, \{a,b\}$ \\
  $c$-arrow & $A_c$ & $\{a,c\}, \{b,c\}$  \\
  $a$-sickle & $S_a$ & $\{a,b\}, \{a,c\}$ \\
  $c$-sickle & $S_c$ & $\{b,c\}, \{a,c\}$ \\
  clockwise triangle-crossing & $C$ & $\{a,c\}, \{b,c\}, \{a,b\}$  \\
  counterclockwise triangle-crossing & $CC$ & $\{a,c\}, \{a,b\}, \{b,c\}$  \\
   \hline
\end{tabular}
 \label{tab:classify}
\caption{Classification of edges crossing the edges of a triangle
$\Delta = (a,b,c)$}
\end{table}

Suppose that $u$ is outside $\Delta$. Then the other endpoint of $g=
\{u,w\}$ is inside $\Delta$ if $g$ is a needle, a hook, or a
triangle-crossing edge, and $w$ is outside $\Delta$ if $g$ is an
arrow or a sickle, see Fig.~\ref{fig:triangle2dir}. An $a$-arrow and
an $a$-sickle are covered by  $a$, since they are crossed by at
least two edges of $fan(a)$. Similarly, a $c$-arrow and a $c$-sickle
are covered by $c$. A needle $g$ may be covered by $a$ or by $c$ and
there is a preference for $a$ ($c$) if $g$ is before (after) any
 triangle-crossing
edge  according to the order of crossing points on $\{a,c\}$ from
$a$ to $c$. Otherwise, there is an instance of configuration II, as
shown in Fig.~\ref{fig:badneedle}. Accordingly, an $a$-hook
   may be covered by $a$ or by $b$   and the
crossing edges are on or inside $\Delta$ if it is covered by $b$,
since the triangle-crossing edges prevent edges from $b$ outside
$\Delta$ that cross $a$-hooks.

By symmetry, we consider needles, hooks, arrows, and sickles from
the viewpoint of vertex $v$ inside $\Delta$. Then a needle first
crosses $\{a,b\}$ and an $a$-hook first crosses $\{a,c\}$ and the
other endpoint is outside $\Delta$.

\begin{figure}[H]
  \centering
  \subfigure[ ]{
    \parbox[b]{4.5cm}{%
      \centering
      \includegraphics[scale=0.7]{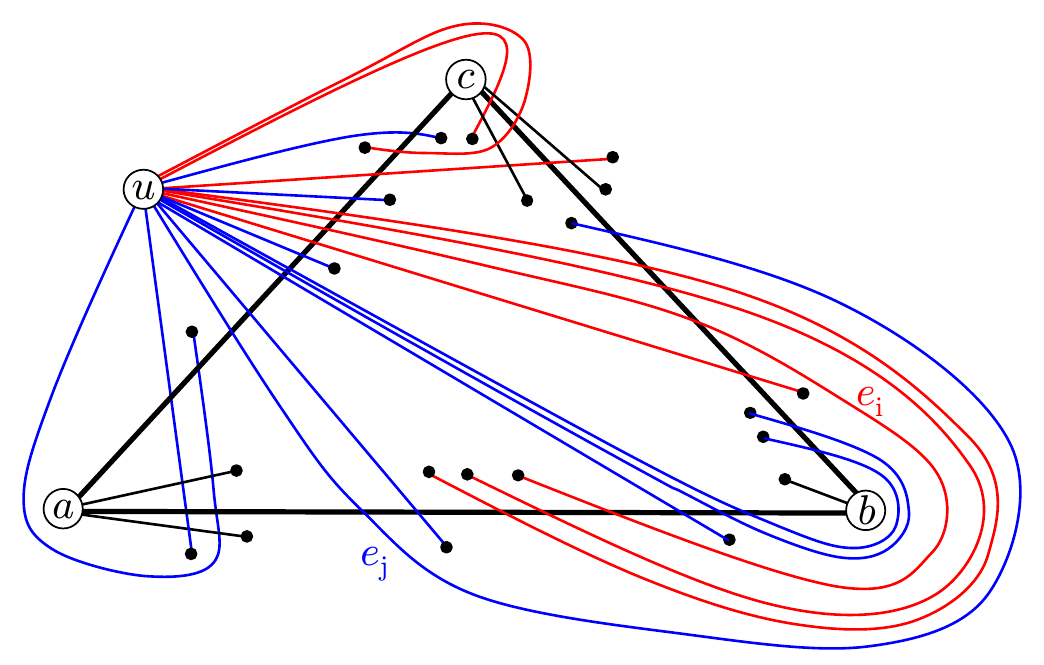}
    }
    \label{fig:triangle2dir}
  }
\\
  \subfigure[ ]{
    \parbox[b]{4.5cm}{%
      \centering
      \includegraphics[scale=0.7]{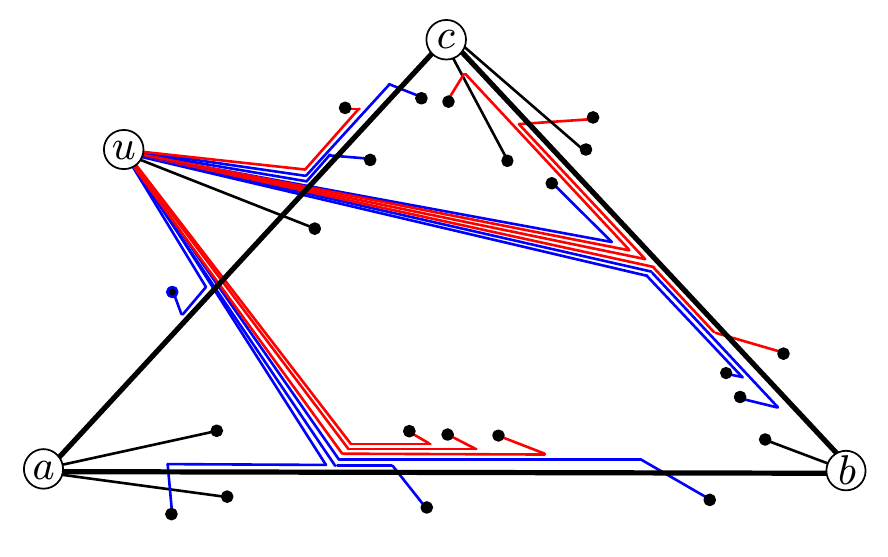}
    }
    \label{fig:triangle2dirRR}
  }
 \caption{Triangle-crossings (a) with clockwise triangle-crossing
 edges, $c$-hooks, $c$-sickles, and $c$-arrows crossing $\{b,c\}$ drawn red and
 counterclockwise triangle crossing edges, $a$-arrows, $a$-hooks and $a$-sickles crossing $\{a,b\}$, drawn blue and
 (b) rerouting the edges along $e_i$ and $e_j$}
  \label{fig:notrianglemultiCC}
\end{figure}

A triangle $\Delta = (a,b,c)$ can be crossed by several
triangle-crossing edges, even in opposite directions, see
Fig.~\ref{fig:triangle2dir}. We say that a triangle-crossing edge
crosses \emph{clockwise}   if it crosses $\{a,c\}, \{b,c\}, \{a,b\}$
cyclicly in this order, and \emph{counterclockwise}   if it crosses
the edges in the cyclic order $\{a,c\}, \{a,b\}, \{b,c\}$.

\begin{lemma} \label{lem:triangle2dir}
Let $\mathcal{E}(G)$ be an adjacency-crossing embedding of a graph
$G$ such that  a triangle $\Delta$ is crossed by triangle-crossing
edges
 in clockwise and in counterclockwise order. Then there is an
adjacency-crossing embedding in which each triangle-crossing edge is
rerouted so that it  crosses only one edge of $\Delta$,  and no new
triangle-crossings are introduced.
\end{lemma}

\begin{proof}
Suppose that the edges of   $\Delta = (a,b,c)$ are crossed by the
edges of a set $X$. If there are at least two triangle-crossing
edges, then there is a vertex $u$ so that $X=fan(u)$. By our
assumption,  $u$ is outside $\Delta$ and $\{a,c\}$ is crossed first.
All other cases are similar. Classify the edges according to
Table~$1$. Choose a clockwise triangle-crossing edge $e_i$ and a
counterclockwise triangle-crossing edge   $e_j$,  and assume that
$e_i$ precedes $e_j$ in clockwise order at $u$. The other case is
similar. Partition the set of needles so that
 $N_1, N_2$ and $N_3$ are the sets of needles before
  $e_i$, between $e_i$ and $e_j$, and after $e_j$ in
clockwise order at $u$. Then $N_3 < e_j < N_2 < e_i < N_1$ according
to the order (of the crossing points) on $\{a,c\}$. Accordingly,
partition the set of counterclockwise triangle-crossing edges into
$CC_l$ and $CC_r$, where $CC_l$ comprises the   edges before $e_i$
and $CC_r = CC-CC_l$ is the set of edges after $e_i$, and partition
the set $C$ into the sets of edges to the left and right of $e_j$.

Then edges of $\Delta$ are crossed by the edges of  $X=N_1 \cup N_2
\cup N_3 \cup H_a \cup H_c \cup A_a \cup A_c \cup L_a \cup L_c\cup C
\cup CC$. Some of these sets may be empty. The edges from these sets
are unordered at $u$. In particular, edges of $C$ and $CC$ may
alternate, needles may appear anywhere, whereas $c$-hooks and
$c$-sickles precede triangle-crossing edges which precede $a$-hooks
and $a$-sickles.

We sort the edges of $X$ in clockwise order at $u$ and reroute them
along $e_i$ and $e_j$ in the following order:

\noindent $S_c < N_1 <    CC_l < H_c < A_c < CC_r < N_2 < C_r < A_a
< H_a < C_l < N_3 < S_a$.

Two edges in a set are ordered by the crossing points with edges of
$\Delta$ so that adjacent edges do not cross one another. The edges
of $S_c$ and $N_1$ are routed along $e_i$ from $u$ to the crossing
point of $e_i$ and $\{a,c\}$, where they make a left turn and follow
$\{a,c\}$. Then the rerouted edge $\tilde{g}$ follows the original
$g$ so that $\tilde{g}$ crosses $\{a,c\}$ if $g$ is a needle. An
edge $\tilde{g}$ first follows $e_i$ to the crossing point with
$\{b,c\}$ if $g \in H_c \cup CC_l \cup A_c \cup CC_r$, then it
follows $\{b,c\}$ and finally $g$. If $g \in H_c \cup CC_l$, then
$\tilde{g}$ makes a left turn and   a right turn for edges in
$CC_r$. Accordingly, edges $\tilde{g}$ make a left or right turn and
cross $\{b,c\}$ if $g$ is an arrow. An edge $\tilde{g}$ may follow
$e_i$ or $e_j$ from $u$ to $\{a,c\}$ or adopts the route of $g$ if
$g \in N_2$ is a needle between the chosen triangle-crossing edges
$e_i$ and $e_j$. Similarly, edges of $C_r, A_a, C_l, N_3$ and $S_a$
are routed along $e_j$ from $u$ to the crossing point with $\{a,b\}$
and $\{a,c\}$, respectively, then along one of these edges, and
finally along the original edge.  For an illustration see
Fig.~\ref{fig:notrianglemultiCC}.

The rerouting saves many crossings. Only arrows cross two edges of
$\Delta$, and needles, hooks and triangle-crossing edges cross
$\{a,c\}$. In fact, each rerouted edge is crossed by a subset of
edges crossing the original one, except if the edge is a hook. This
is due to the fact that triangle-crossing edges are only crossed by
the edges of the triangle. Hence, there are (uncrossed) segments
from $u$ to $\{a,c\}$ and from $\{a,c\}$ to $\{b,c\}$ and $\{a,b\}$,
respectively.
  In the
final part, $\tilde{g}$ coincides with $g$ and adopts the edge
crossings from $g$. In consequence, $\tilde{g}$   crosses only
$\{a,c\}$ if $g$ is a triangle-crossing edge.
 If $g$ is a $c$-hook, then
the crossing with edge $\{b,c\}$  is replaced by a crossing with
$\{a,c\}$ and crossings with edges of $fan(c)$ outside $\Delta$ are
avoided. The replacement is feasible.  A $c$-hook cannot be covered
by $b$, since a further crossing edge $\{b,d\}$ must cross a
clockwise triangle-crossing edge, which is excluded. Hence,
$\tilde{g}$ is crossed by edges of $fan(c)$, and each edge $h$
crossing $\tilde{g}$ is in $fan(u)$. Similarly, edge $\{a,b\}$ can
be replaced by $\{a,c\}$ at $a$-hooks. The other rerouted edges
adopt the crossings from the final part, so that   new
triangle-crossings cannot be introduced. Topological simplicity is
preserved, since the bundle of edges is well-ordered, and two edges
cross at most once, since there are segments from $u$ to $\{a,c\}$
and between $\{a,c\}$ and $\{b,c\}$ and $\{a,b\}$, respectively.

In consequence,  triangle-crossings of   $\Delta$ are avoided, there
are no new triangle-crossings, and the obtained embedding is
adjacency-crossing.
\qed
\end{proof}

The rerouting technique of Lemma \ref{lem:triangle2dir} widely
changes the order of the edges of $fan(u)$ and it avoids many
crossings.   It is possible to restrict the rerouting to
triangle-crossing edges so that they cross only a single edge of the
triangle. Therefore consider two consecutive crossing points of
clockwise triangle crossing edges or $c$-arrows and $\{b,c\}$, and
reroute the counterclockwise crossing edges crossing $\{b,c\}$ in
the sector along one of the bounding edges. Accordingly, proceed
with clockwise triangle-crossing edges and sectors of $\{a,b\}$.
Thereby hooks, sickles and arrows remain unchanged.
 \\

From now on, we   assume that all triangle-crossing edges   cross
 clockwise. We wish to reroute them along
 an $a$-arrow, $a$-hook or $a$-sickle if such an edge exists.
 This is doable, but we must take a detour if
 the edge is covered by $b$ or $c$.

\begin{lemma} \label{lem:trianglearrow}
Suppose there is an adjacency crossing embedding $\mathcal{E}(G)$
and a triangle $\Delta$   is crossed by clockwise triangle-crossing
edges. If there are an
 $a$-hook, an $a$-arrow or an $a$-sickle, then some edges are
 rerouted so that
 $\tilde{g}$ crosses only one edge of $\Delta$ if $g$ is a triangle-crossing edge of $\Delta$,
 and there are no new triangle-crossings.
\end{lemma}

\begin{proof}
Our target is edge $\{a,b\}$ of $\Delta=(a,b,c)$, where the crossing
edges are ordered from $a$ to the left to $b$. Then $a$-hooks and
$a$-sickles are to the left of all triangle-crossing edges, whereas
$a$-arrows are interspersed. Edge $\{a,b\}$ is covered by $u$.

Let $f= \{u,w\}$ be the rightmost edge among all $a$-hooks,
$a$-arrows, and $a$-sickles. First, if $f$ is an $a$-hook, then
reroute all  edges $g$ crossing $\{a,b\}$ to the right of $f$ in a
bundle from $u$ to $\{a,b\}$ along the outside of $f$, see
Fig.~\ref{fig:trianglehook}. Since $f$ is rightmost, edge $g$ is
triangle-crossing. Then $\tilde{g}$ makes a right turn and follows
$\{a,b\}$ and finally it follows $g$. Thereby, $\tilde{g}$ crosses
$\{a,b\}$. Let $F$ be the set of edges
 in the sector between $\{a,b\}$ and
$\{a,c\}$ that cross $f$, i.e., outside $\Delta$. Then $\tilde{g}$
is crossed by the edges of $F$ and also by $\{a,b\}$. Each crossing
edge is in $fan(a)$ and is uncovered or covered by $u$. It cannot be
covered by the other endpoint $w$ of $f$, since $w$ is inside
$\Delta$ and any edge $\{w,w'\}$ crossing an edge $\{a,d\} \in F$
must cross $\{a,b\}, \{a,c\}$ or a triangle-crossing edge, which is
excluded, since it enforces an independent crossing. Thus
$\tilde{g}$ is only crossed by edges  of $fan(a)$, and $\tilde{g}$
can be added to the fan of edges of $fan(u)$  that cross such edges.
Hence, all introduced crossings are fan-crossings, as
Fig.~\ref{fig:trianglehookR} shows.

We would like to proceed accordingly if $f$ is an $a$-sickle and
reroute triangle-crossing edges along the outside of $f$ from $u$ to
$\{a,b\}$. However, $f$ may be crossed by   edges $\{a,d\}$ that are
covered by $w$, as shown in Fig.~\ref{fig:trianglesickle}. Then a
rerouted edge along  $f$ introduces an independent crossing. We take
another path.

Let the $a$-sickle $f= \{u,w\}$   cross  $\{a,b\}$ in $p_1$ and
  $\{a,c\}$ in  $p_2$, see Fig.~\ref{fig:trianglesickle}.
Let $H$ be the set of edges that   cross $\{a,c\}$ between the first
triangle-crossing edge $e_1$ and $f$ including $f$. Now  we reroute
  all edges $h \in H$ and all triangle-crossing edges $g$ so that
they first follow
  $e_1$ from $u$ to $\{a,c\}$, then
$\{a,c\}$, where  the edges $\tilde{h}$ branch off and and follow
$h$. If $g$ is a triangle-crossing edge, then $\tilde{g}$ crosses
$\{a,c\}$ at $p_2$, and then follows $f, \{a,b\}$, and finally $g$,
see Fig.~\ref{fig:trianglesickleR}.

The rerouted edges are uncrossed from $u$ to their crossing point
with $\{a,c\}$. Hence, each edge  $\tilde{h}$ is crossed by a subset
of edges that cross $h$ for $h \in H$. Let $F$ be the set of edges
crossing $f$ in the sector between $p_1$ and $p_2$. Since $f$ is
covered by $a$, these edges are incident to $a$. Now $\tilde{g}$ is
crossed by $\{a,c\}$ and by the edges of $F$ if $g$ is
triangle-crossing, so that $\tilde{g}$ is crossed by edges of
$fan(a)$. Each edge $h \in F$ is in $fan(u)$, since it crosses $f=
\{u,w\}$ and it cannot be covered by $w$. Otherwise, it must be
crossed by another edge $\{w, w'\}$. However, $w$ is outside
$\Delta$ and $\{w, w'\}$ must cross $\{a,c\}$ or $\{a,b\}$ or a
triangle-crossing edge, which introduces an independent crossing.
Hence, $\tilde{g}$ can be added to the fan of edges at $u$ that
cross $h$ so that there is a fan-crossing.

We proceed similarly
 if $f = \{u,w\}$ is an $a$-arrow, see Fig.~\ref{fig:tribadarrowX}.
 Reroute all edges $g$ that cross $\{a,c\}$
 to the right of the leftmost  triangle-crossing edge $e_1$ including $e_1$. Then $g$ is triangle-crossing
 or an $a$-arrow. Route $\tilde{g}$ from $u$ to $\{a,c\}$ along the
 first edge that crosses $\{a,c\}$ and is covered by $c$, then along
 $\{a,c\}$ to the crossing point with $f$, then along $f$ and
 finally along $g$. Then there is a segment from $u$ to the crossing
 with $\{a,c\}$. In the sector between $\{a,c\}$ and $\{a,b\}$, $\tilde{g}$ is
 crossed by the edges of $fan(a)$ that cross $f$ in this sector. If
 $g$ is a triangle-crossing edge, then $\tilde{g}$ is not crossed by
 further edges, whereas $\tilde{g}$  adopts the crossings with further
 edges incident to $a$ outside $\Delta$ if $g$ is an $a$-arrow.

 Now, $\tilde{g}$ is crossed by a subset of edges that cross $g$ if $g$ is an
 $a$-arrow, since $f$ is the rightmost $a$-arrow. If $g$ is a
 triangle-crossing edge, then the edges crossing $\tilde{g}$ are
 incident to $a$, and each crossing edge is incident to $u$. It
 cannot be incident to or covered by the other endpoint $e$ of $f$,
 since $w$ is outside $\Delta$ and the edges crossing $\tilde{g}$
 are inside, and and no further edge $\{w,w'\}$ with $w' \neq u$ can
 cross $\{a,b\}, \{a,c\}$, or a triangle-crossing edge. Hence, there
 is a fan-crossing,   $\tilde{g}$ crosses only one edge of $\Delta$ if
$g$ is triangle-crossing, and there are no new triangle-crossings.
\qed
\end{proof}

\begin{figure}
  \centering
  \subfigure[ ]{
    \parbox[b]{4.5cm}{%
      \centering
      \includegraphics[scale=0.6]{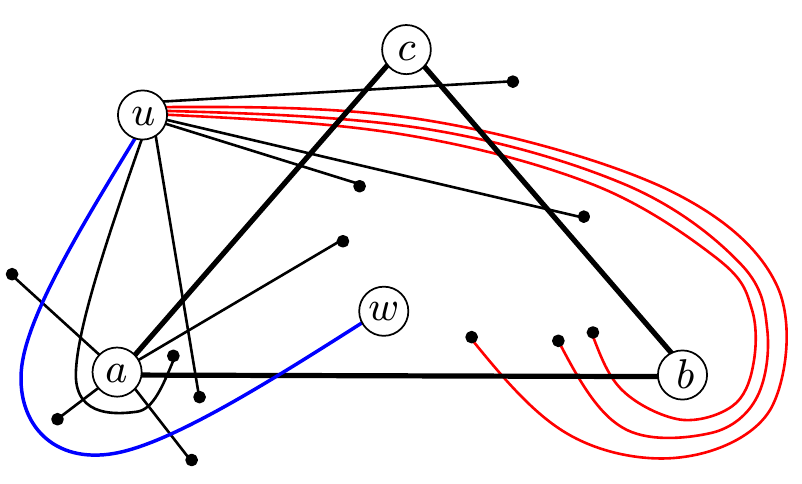}
    }
    \label{fig:trianglehook}
  }
  \hfil
  \subfigure[ ]{
    \parbox[b]{4.5cm}{%
      \centering
      \includegraphics[scale=0.6]{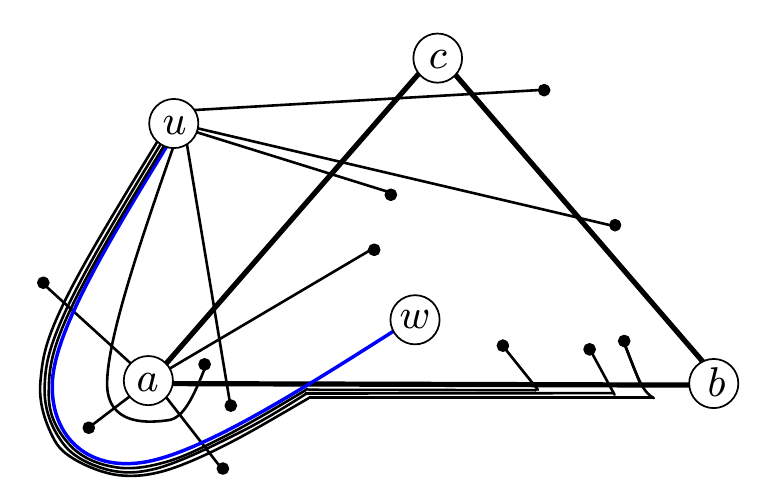}
    }
    \label{fig:trianglehookR}
  }
 \caption{(a) An $a$-hook (drawn blue and dashed) and triangle-crossing edges which (b) are rerouted along the $a$-hook.}
  \label{fig:trihook}
\end{figure}

\begin{figure}
  \centering
  \subfigure[ ]{
    \parbox[b]{4.5cm}{%
      \centering
      \includegraphics[scale=0.6]{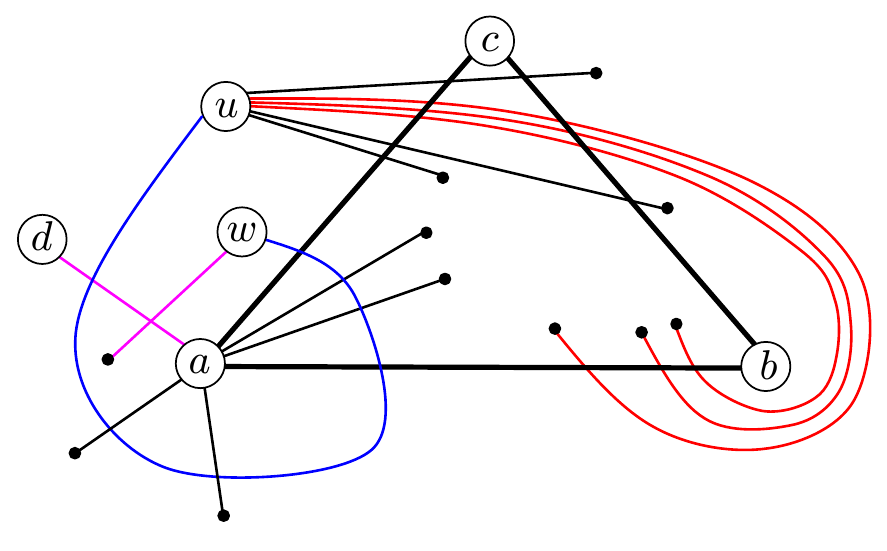}
    }
    \label{fig:trianglesickle}
  }
  \hfil
  \subfigure[ ]{
    \parbox[b]{4.5cm}{%
      \centering
      \includegraphics[scale=0.6]{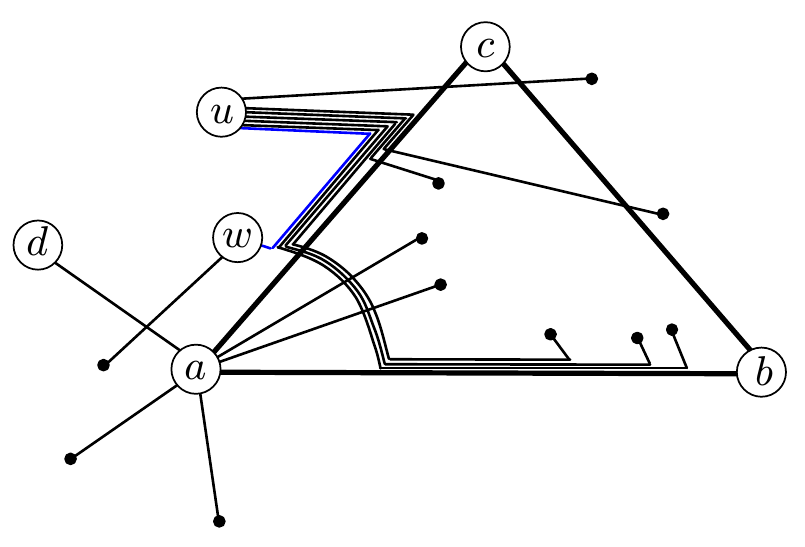}
    }
    \label{fig:trianglesickleR}
  }
 \caption{  An $a$-sickle and triangle-crossing edges (a) before and   (b) after the edge rerouting.}
  \label{fig:sickle}
\end{figure}

\begin{figure}
  \centering
  \subfigure[ ]{
    \parbox[b]{4.5cm}{%
      \centering
      \includegraphics[scale=0.6]{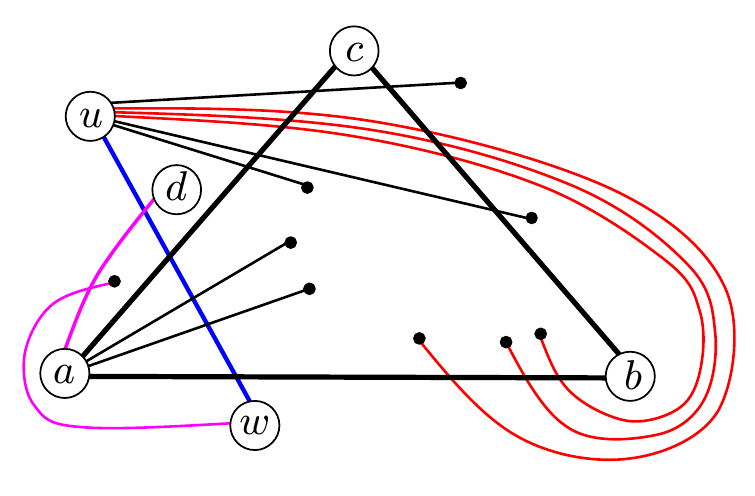}
    }
    \label{fig:triarrow}
  }
  \hfil
  \subfigure[ ]{
    \parbox[b]{4.5cm}{%
      \centering
      \includegraphics[scale=0.6]{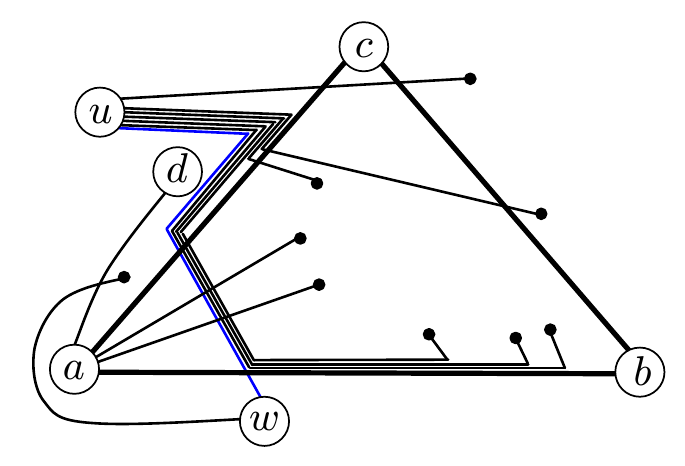}
    }
    \label{fig:triarrowR}
  }
 \caption{An $a$-arrow and triangle-crossing edges (a) before and   (b) after the edge
 rerouting.}
  \label{fig:tribadarrowX}
\end{figure}


The existence of an $a$-hook, $a$-sickle  or $a$-arrow implies that
edge $\{a,b\}$ is covered by $u$. By symmetry,   we can reroute all
triangle-crossing edges, if there are $a$-hooks, $a$-sickles or
$a$-arrows from the viewpoint of vertex $v$ inside $\Delta$. Then
$\{a,c\}$ is covered by $v$. For example, an arrow from $v$ first
crosses $\{a,b\}$ and then $\{b,c\}$ so that vertex $b$ is enclosed
and triangle-crossing edges are rerouted along the outer side of the
arrow. It remains to consider the case without such edges. Then
there are only triangle-crossing edges, needles (from $u$ and from
$v$),   $c$-hooks, $c$-arrows, and $c$-sickles.

\begin{lemma} \label{lem:trianglecovered}
Suppose there is an adjacency crossing embedding $\mathcal{E}(G)$
and a triangle $\Delta = (a,b,c)$   is crossed by  clockwise
triangle-crossing edges. If there are no
 $a$-hooks,   $a$-arrows and  $a$-sickles
and edges $\{a,c\}$ and $\{b,c\}$ are not covered by $v$, then edge
$\ell=\{a,b\}$ can be rerouted so that  $\tilde{\ell}$ does not
cross the rerouted edge,
 and there are no new triangle-crossings.
 Similarly, reroute $\{a,c\}$ if $\{b,c\}$ is not covered by $u$ and
 there are no $a$-hooks,   $a$-arrow and  $a$-sickles from the
 viewpoint of $v$.
\end{lemma}

\begin{proof}
Besides one or more clockwise triangle-crossing edges there are only
needles, $c$-hooks, $c$-arrows and $c$-sickles. We cannot route the
triangle-crossing edges along the edges of $\Delta$, since vertices
$a$ and $b$ may be incident to ``fat edges'', that are explained in
Section \ref{sect:fanplanar}, and prevent a bypass. Therefore, we
reroute $\{a,b\}$. Similarly, we reroute $\{a,c\}$ if $\{a,b\}$ and
$\{b,c\}$ are not covered by $u$, and both ways may be possible.

If $\{u,b\}$ is an edge of $G$, then it crosses $\{a,c\}$ and we
take $f=\{u,b\}$; 
otherwise let $f$ be the first edge crossing both
$\{a,c\}$ in $p_1$ and $\{b,c\}$ in $p_2$. Then $f$ is covered by
$c$ and  is a triangle-crossing edge or a $c$-arrow. There is a
segment from $u$ to $p_1$, from $p_1$ to $p_2$, and from $p_2$ to
$b$. Other edges incident to $c$ cannot cross $f$, since $f$ is
triangle-crossing or is protected from $c$ by a triangle-crossing
edge, and the final part along $\{b,c\}$ is uncrossed, because $f$
is the first edge crossing $\{b,c\}$ from $b$.

Reroute  $\ell = \{a,b\}$  so that $\tilde{\ell}$ first follows
$\{a,c\}$ from $a$ to $p_1$, then $f$ to $p_2$ and finally $\{b,c\}$
to $b$. If $f= \{u,b\}$, then $p_2$ and $b$ coincide. 
Let $N$ be the set of edges crossing $\{a,c\}$ in the segment from
$a$ to $p_1$. Then $N$ consists of needles so that $N = N_c \cup
N_a$, where a needle $n \in N_c$ is covered by $c$ and a needle $n
\in N_a$ is uncovered or covered by $a$. The needles in $N_c$ cross
$\{a,c\}$ before the needles of $N_a$. In fact, if an edge $\{x,y\}$
other than $\{a,c\}$ crosses a needle $n \in N$, then $\{x,y\}$ is
outside $\Delta$ if $n \in N_c$. If $\{x,y\}$ crosses $n$ inside
$\Delta$, then  $n \in N_a$, since further edges incident to $c$
cannot enter the interior of $\Delta$ below the triangle-crossing
edges.

Now $\tilde{\ell}$ is crossed by the edges of $N$. Note that there
are no crossings of $\tilde{\ell}$ in the second part along $f$ and
in the third part along $\{b,c\}$.  Since the edges of $N$ are
incident to $a$, $\tilde{\ell}$ is crossed by edges $fan(a)$. In
return, consider an edge $h$ crossing some needle $n = \{u,w\} \in
N$. Then $n$ and may be covered by $a$ or by $c$ so that $h=
\{a,d\}$ or $h=  \{d,d\}$. If $h$ is not covered by $c$, we are
done, since we can add $\tilde{\ell} = \{a,b\}$ to the fan of edges
of $fan(a)$ crossing $n$.

However, there is a conflict if $n$ is covered by $c$, as shown in
Fig.~\ref{fig:badneedle}.
Then there are needles $\{u,w_1\}, \ldots, \{u,w_s\}$ and edges
$\{c, z_1\},\ldots, \{c,z_t\}$ for some $s,t \geq 1 $ so that each
$\{u,w_i\}$ is crossed by some  $\{c, z_j\}$.

We resolve the conflict by  rerouting the  needles  in advance, so
that needles of $N_c$ are no longer covered by $c$, see
Fig.~\ref{fig:badneedleR}. Reroute each needle $\tilde{n}$   from
$u$ to $p_1$ along $f$, then along $\{a,c\}$, and finally along $n$.
Then there is a segment from $u$ to the crossing point with
$\{a,c\}$ so that $\tilde{n}$ is only crossed by a subset of edges
that cross $g$. Thereafter, there are no needles covered by $c$, and
we are done.
\qed
\end{proof}

\begin{figure}
  \centering
  \subfigure[ ]{
    \parbox[b]{4.5cm}{%
      \centering
      \includegraphics[scale=0.6]{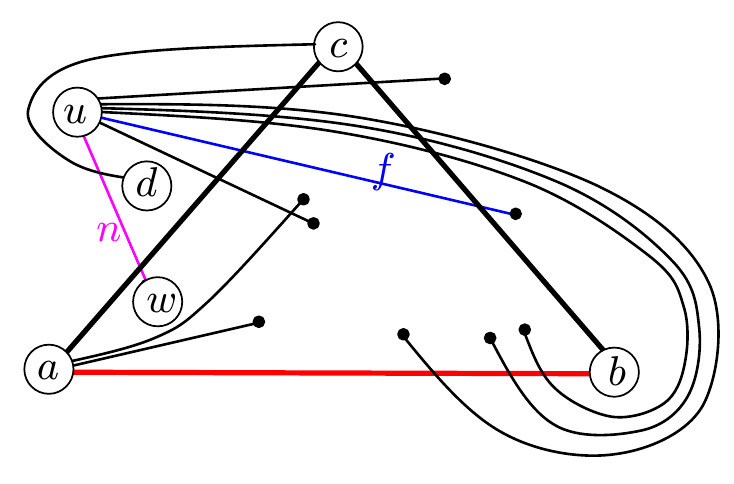}
    }
    \label{fig:badneedle}
  }
  \hfil
  \subfigure[ ]{
    \parbox[b]{4.5cm}{%
      \centering
      \includegraphics[scale=0.6]{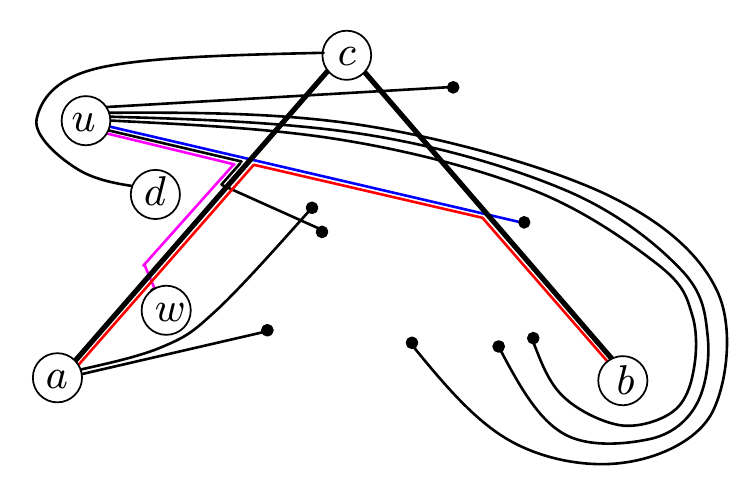}
    }
    \label{fig:badneedleR}
  }
 \caption{A triangle-crossing (a) with a needle covered by vertex
 $c$ that introduces configuration II
 and an edge rerouting that avoids triangle-crossing edges.}
  \label{fig:routebase}
\end{figure}

We can now show that triangle-crossings can be avoided.

\begin{theorem} \label{thm:trianglecrossing}
Every adjacency-crossing graph is fan-crossing.
\end{theorem}

\begin{proof}
Let $\mathcal{E}(G)$ be an adjacency-crossing embedding of a graph
$G$ and suppose that there are triangle crossings. We   remove them
one after another   and first  consider  all triangles
 with triangle-crossing edges in both directions
(Lemma \ref{lem:triangle2dir}), then the triangles with $a$-hooks,
$a$-arrows or $a$-sickles (Lemma \ref{lem:trianglearrow}), and
finally those without such edges (Lemma \ref{lem:trianglecovered}).
 Each step
removes a crossed triangle and does not introduce new ones. Hence,
the resulting embedding is fan-crossing.
\qed
\end{proof}

\section{Fan-Crossing and Fan-Planar  Graphs} \label{sect:fanplanar}

In this Section  we assume that embeddings are fan-crossing so that
independent crossings and triangle-crossings are excluded.
Fan-planar embeddings also exclude configuration II
\cite{ku-dfang-14}. An instance   of configuration II consists of
the fan-crossing embedding of a subgraph $C$ induced by the vertices
of an edge $e=\{u,v\}$  and of all edges $\{t,w\}$ crossing $e$,
where $e$ is crossed from both sides, as shown in
Fig.~\ref{fig:conf2}. We call $e$ the \emph{base} and its crossing
edges the \emph{fan} of $C$, denoted $fan(C)$. Since $e$ is crossed
from both sides, it it crossed at least twice, and therefore  it is
covered by $t$.
 It may be crossed by more than two edges. Hence, an edge is the base of
 at most one configuration, but a base
 may be in the fan of another configuration.
Each edge $g$   of $fan(C)$ is uncovered or is covered by exactly
one of $u$ and $v$. It may cross several base edges so that it is
part of several configurations.
%
An edge of $fan(C)$ is said to be \emph{straight} if it crosses $e$
from the left and \emph{curved} if it crosses $e$ from the right.
Then an instance of configuration II has at least a straight and a
curved edge. Moreover, exactly one of $u$ and $v$ is inside a cycle
with edge segments of a curved edge, the base, and a straight edge.
For convenience, we assume that $u$ is inside the
  cycle and curved edges are \emph{left  curves}. Right curves enclose $v$ and both left and right curves
  are possible. However, if there are left and right curves, then curves in one direction can
  be rerouted.

For convenience, we augment  the embedding and assume that for every
instance $C$ of configuration II there are edges $\{t,u\}$ and
$\{t,v\}$. If these edges do not exist,  they can be added.
Therefore, route $\{u,t\}$
 along the first left  curve $f$    from $u$ to the
  first crossing point with an edge $g$ of $fan(u)$  and then
along $g$. Then $f$ is uncovered or covered by $u$ and $\{t,u\}$ is
uncrossed, or $f$ is covered by $v$ and $\{t,u\}$ is covered by $v$
or is uncovered. Accordingly, $\{t,v\}$ follows the rightmost edge
crossing $e$ and the first crossed edge of $fan(v)$. The case with
right curves is similar. Hence, we can assume that there is a
triangle $\Delta = (t,u,v)$ associated with $C$.


There are some  cases in which configuration II can be avoided by an
edge rerouting.  A special one has been used in Lemma
\ref{lem:trianglecovered} in which the straight edge is crossed by a
triangle-crossing edge. However, there is a case in which
configuration II is unavoidable.

\begin{lemma} \label{lem:conf2covercurve}
If  a straight edge $s$ of an instance $C$ of configuration II is
uncovered or is covered by $u$, then the left curves $g$ to the left
of $s$  can be rerouted so that $\tilde{g}$ does not cross the base.
  The edge rerouting does not introduce new instances of
configuration II.
 \end{lemma}

\begin{proof}
We reroute each edge $g$ to the left of $s$ so that $\tilde{g}$
first follows $s$ from $t$ to the crossing point with the first edge
$f$  of $fan(u)$ that crosses both $g$ and $s$. Then $\tilde{g}$
follows $f$ and finally $g$. If $g$ is  a straight edge, then $f =
\{u,v\}$, which is crossed.
 See
Fig.~\ref{fig:conf2leftsR} for an illustration. If $g$ is a left
curve, then $\tilde{g}$ is only crossed by the edges of $fan(u)$
that cross $s$ in the sector between $\{u,t\}$ and $f$,  and by the
edges that cross $g$ in the sector from $f$ to the endpoint. All
edges are in $fan(u)$ and $\{u,v\}$ is not crossed by $\tilde{g}$.
Each edge $h$ that is crossed by $\tilde{g}$ is crossed only once,
since $f$ is the first edge crossing $g$ and $s$. If $h \in fan(u)$
is crossed by $\tilde{g}$ and $g$ and $h$ do not cross, then $h$
crosses $s$ and $h$ is a straight edge for $\tilde{g}$. If there is
a curved edge $\{u,w\}$ crossing $\tilde{g}$, then $\{u,w\}$ is also
a curved edge for $s$. Hence, $\tilde{g}$ can be added to that
instance of configuration II.
 If $g$ is a straight edge,
then $\tilde{g}$ is crossed by a subset of edges that cross $g$,
since each edge of $fan(u)$ crossing $s$ in the sector between
$\{u,t\}$ and $\{u,v\}$ must cross $g$. Hence there are no more edge
crossings and instances of configuration II.
\qed
\end{proof}

\begin{figure}[h]
  \centering
  \subfigure[ ]{
    \parbox[b]{6.0cm}{%
      \centering
      \includegraphics[scale=0.7]{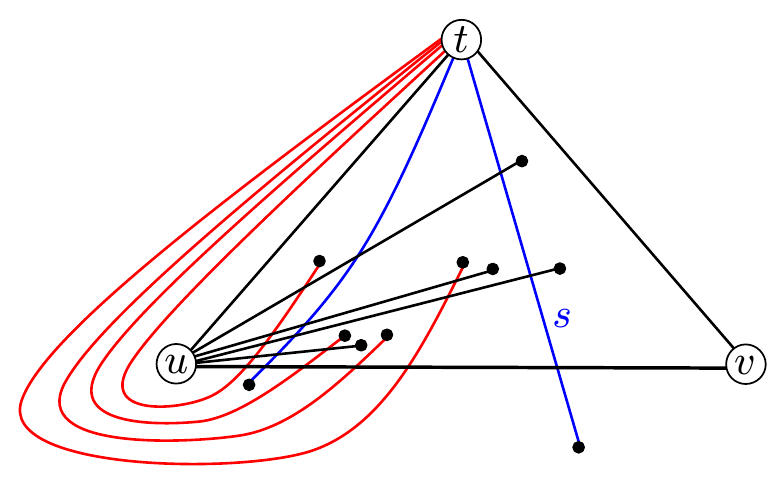}
    }
    \label{fig:conf2left}
  }
  \hfil
  \subfigure[ ]{
    \parbox[b]{5.5cm}{%
      \centering
      \includegraphics[scale=0.7]{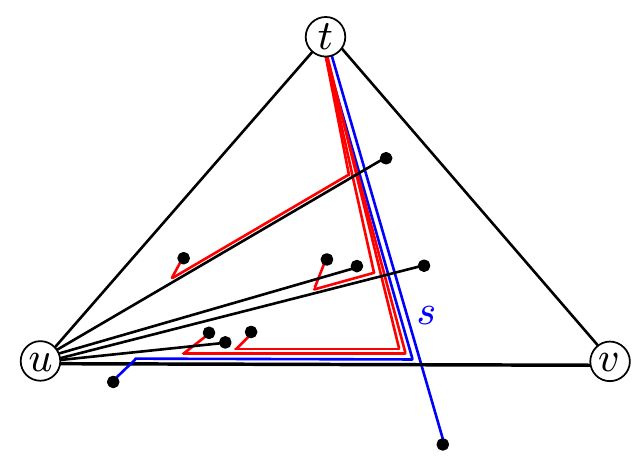}
    }
    \label{fig:conf2leftR}
  }
 \caption{  An instance of configuration II with (a) a straight edge $s$  covered by $u$
 and left curves to its left and   (b) rerouting the edges crossing
 $\{u,v\}$ to the left of $s$. }
  \label{fig:conf2leftsR}
\end{figure}

In consequence, we can remove instances of configuration II in which
there are left curves, right curves and straight edges, since Lemma
\ref{lem:conf2covercurve} either applies to the left or to the right
curves. Lemma \ref{lem:conf2covercurve} cannot be used if left
curves are to the right of straight edges, since the left curves may
be covered by $v$ and the straight edges by $u$. Then configuration
II may be unavoidable using a construction similar to the one of
Theorem \ref{thm:notfanplanar}.\\

A left curve $g= \{t,x\}$ is \emph{semi-covered} by $u$ if it is
only crossed by an
 edge $\{u,w\}$  in the sector between $\{u,t\}$ and $\{u,v\}$.
 Thus the crossing edge is inside the triangle $\Delta = (t,u,v)$.
   Accordingly, a straight edge $h= \{t,y\}$ is \emph{semi-covered} by $v$
if each edge $\{v,w\}$ with $w\neq u$ crosses  $h$ in the sector
between $\{v,t\}$ and $\{v,u\}$, i.e., outside $\Delta$.

A semi-covered edge is covered, but not conversely. A covered left
curve that is not semi-covered is crossed by edges of $fan(u)$ in
the sector between $\{t,v\}$ and $\{t,u\}$ in clockwise order, i.e.,
outside the triangle $(t,u,v)$. Similarly, a semi-covered straight
edge may be crossed by edges of $fan(v)$ inside the triangle. Thus a
semi-covered left curve consists of a segment from $u$ to the
crossing with $\{u,v\}$ and a semi-covered straight edge is
uncrossed inside $\Delta$. These segments are good for routing other
edges.

\begin{lemma} \label{lem:oonf2semicovered}
If there is a semi-covered straight (curved) edge, then all  curved
(straight) edges
  can be rerouted such that they do not cross the base, so that
  configuration II is avoided.
\end{lemma}

\begin{proof}
We proceed as in Lemmas \ref{lem:triangle2dir} and
\ref{lem:trianglearrow} and reroute all straight and curved edges in
a bundle along the semi-covered edge $f$ from $t$ to the base
$\{u,v\}$, where they make a left or right turn, follow the base and
finally their original. If $f$ is straight (curved), then the curved
(straight) edges do not cross the base. Each rerouted edge
$\tilde{g}$ is only crossed by a subset of edges that cross $g$,
since the   part of $\tilde{g}$ is uncrossed until it meets $g$.
\qed
\end{proof}

Next, we construct   graph $M$ in which  configuration II is
unavoidable. Graph $M$ has  fat   and ordinary edges. A \emph{fat
edge} consists of   $K_7$. In fan-crossing graphs, a fat edge plays
the role of an edge in planar graphs. It is impermeable to any other
fat or ordinary edge. This observation is due to Binucci et al.
\cite{bddmpst-fan-15} who proved the following:

\begin{lemma} \label{lem:fatedge}
For every fan-crossing embedding of $K_7$ and every pair of vertices
$u$ and $v$ there is a path of segments in which at least one
endpoint is a crossing point. Thus, each pair of vertices is
connected if the uncrossed edges are removed.
\end{lemma}

There are (at least) three fan-crossing embeddings of $K_7$ with
$K_5$ as in Figs.~\ref{fig:allK5}(a-c) and two vertices in the outer
face, see Fig.~\ref{fig:allK7}. The embeddings in
Figs.~\ref{fig:allK5}(d) and \ref{fig:allK5}(e) cannot be extended
to a fan-crossing embedding of $K_7$ by adding two vertices in the
outer face.

\begin{figure}
  \centering
  \subfigure[ ]{
    \parbox[b]{3.3cm}{%
      \centering
      \includegraphics[scale=0.5]{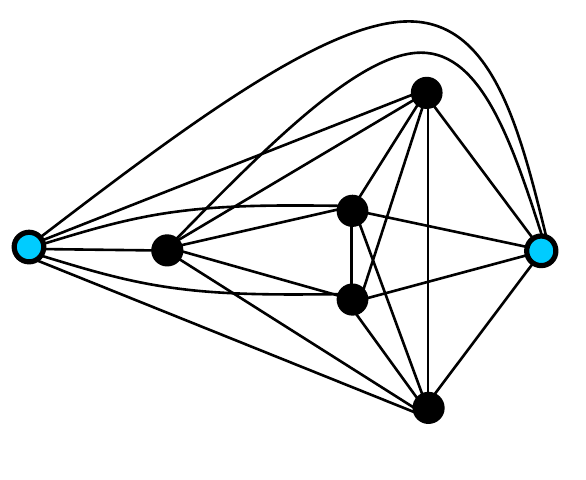}
      }
    \label{fig:K7D1}
  }
  \hfil
  \subfigure[ ]{
    \parbox[b]{3.3cm}{%
      \centering
      \includegraphics[scale=0.5]{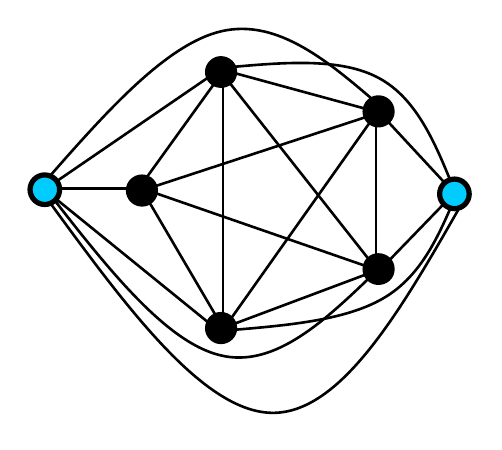}
    }
    \label{fig:K7D2}
  }
  \hfil
  \subfigure[ ]{
    \parbox[b]{3.3cm}{%
      \centering
      \includegraphics[scale=0.5]{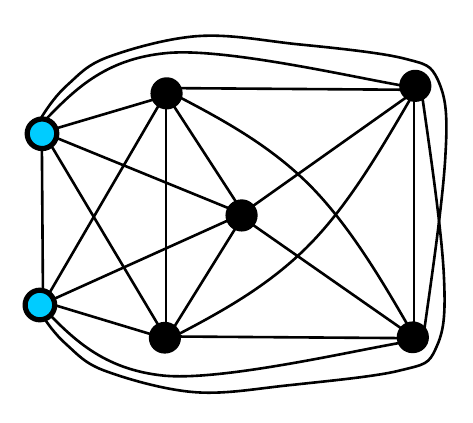}
    }
    \label{fig:K7D3}
  }
 \caption{Different fan-crossing embeddings of $K_7$ that are obtained from different embeddings
 of $K_5$ by adding two vertices in the outer face}
  \label{fig:allK7}
\end{figure}

\begin{theorem} \label{thm:notfanplanar}
There are fan-crossing graphs that are not fan-planar. In other
words, configuration II is unavoidable.
\end{theorem}

\begin{proof}
Consider   graph $M$ from Fig.~\ref{fig:graphM} with  fat edges
representing $K_7$  and ordinary ones. Up to the embedding of the
fat edges, graph $M$ has a unique fan-crossing embedding. This is
due to the following fact.

There is a fixed outer frame consisting of two 5-cycles with
vertices $U = \{t', v', y',a',b', t,v,y,a,b\}$ and fat edges. If fat
edges are contracted to edges or regarded as such, this subgraph is
planar and 3-connected and as such has a unique planar embedding. By
a similar reasoning,   $M[U]$ has  a fixed fan-crossing embedding up
to the embeddings of $K_7$. There are two disjoint 5-cycles, since
fat edges do not admit a penetration by any other edge. Hence, the
edges $\{t,y\}$ and $\{b,v\}$ must be routed inside a face of the
embedding of $M[U]$, and they cross. Consider the subgraph
$M[t,s,u,w,x,z]$ restricted to fat edges. Since vertex $t$ is in the
outer frame, it admits four fan-crossing embeddings with outer face
$(t,u,x,w,z),(t,u,x,z), (t,u,s)$, and $(t,s,z)$, respectively. But
the   edges $\{u,a\}, \{u,b\}, \{v,w\}$ and $\{v,z\}$ exclude the
latter three embeddings, since the edges on the outer cycle are fat
edges and do not admit any penetration by another edge.

Edge $\{u,a\}$ cannot cross $\{t,y\}$, since the latter is crossed
by $\{v,z\}$. Hence, $\{t,y\}$ is crossed by $\{w,v\}$ and
$\{z,v\}$. Finally, edge $\{t,x\}$ must cross $\{u,w\}$. It cannot
cross  $\{v,z\}$ without introducing an independent crossing. Hence,
it must cross $\{u,a\}, \{u,b\}, \{u,v\}$ and $\{u,w\}$.

 Modulo the
embeddings of $K_7$,  every fan-crossing embedding is as shown in
Fig.~\ref{fig:graphM} in which $\{u,v\}$ is crossed by $\{t,x\}$
from the right and by $\{t,y\}$ from the left and thus is
configuration II. Hence, graph $M$ is fan-crossing and not
fan-planar.
\qed
\end{proof}

\begin{figure}
  \centering
    \includegraphics[scale=0.7]{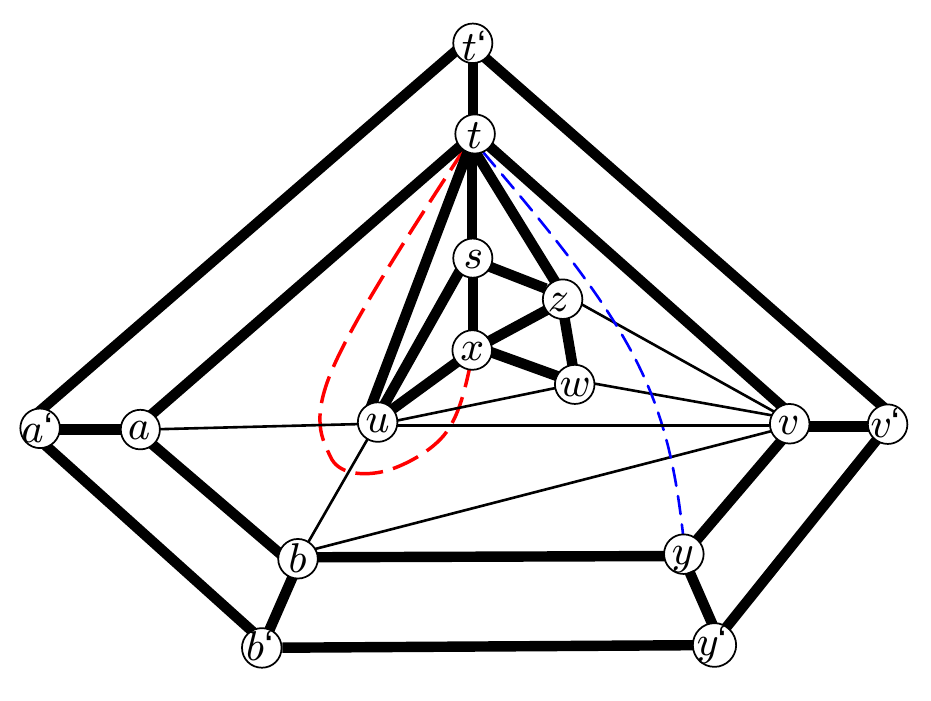}
  \caption{Graph $M$ with fat edges representing $K_7$ and an unavoidable configuration II}
  \label{fig:graphM}
\end{figure}

Theorems \ref{thm:trianglecrossing} and \ref{thm:notfanplanar} solve
a problem of my recent paper on beyond-planar graphs
\cite{b-FOL-17}. Let FAN-PLANAR, FAN-CROSSING, and ADJ-CROSSING
denote the classes of fan-planar, fan-crossing, and
adjacency-crossing graphs. Then Theorems \ref{thm:trianglecrossing}
and \ref{thm:notfanplanar} show:

\begin{corollary}
FAN-PLANAR $\subset$ FAN-CROSSING $=$ ADJ-CROSSING.
\end{corollary}

Kaufmann and Ueckeredt \cite{ku-dfang-14} have shown that fan-planar
graphs of size $n$ have at most $5n-10$ edges, and they posed the
density of fan-crossing  and adjacency-crossing graphs as an open
problem.


\begin{theorem}
For every adjacency-crossing graph $G$ there is a fan-planar graph
$G'$ on the same set of vertices and with the same number of edges.
\end{theorem}

\begin{proof}
By Theorem \ref{thm:trianglecrossing} we can restrict ourselves to
fan-crossing graphs. Let $\mathcal{E}(G)$ be a fan-crossing
embedding of $G$
and suppose there is    an instance of configuration II in which the
base $\{u,v\}$ is crossed by $\{t, x\}$ from the right and by
$\{t,y\}$ from the left, or vice-versa.

Augment $\mathcal{E}(G)$ and add edges $\{u,w\}$ if they are
fan-crossing and do not cross both $\{t, x\}$ and $\{t, y\}$, and
similarly, add  $\{v,w\}$.
Consider the cyclic order of edges or neighbors of $u$ and $v$
starting at  $\{u,v\}$ in    clockwise order. Let $a$ and $b$ be the
vertices encountered first. Vertices $a$ and $b$ exist, since $a$
precedes $x$  and $b$ precedes $y$, where $x=a$ or $b=y$ are
possible. Then $a$ and $b$ are both incident to both $u$ and $v$ and
there are two faces $f_1$ and $f_2$ containing a common segment of
$\{u,v\}$ and $a$ and $b$, respectively, on either side of
$\{u,v\}$. Otherwise, further edges can be added that are routed
close to $\{u,v\}$ and are crossed either by edges of $fan(t)$ that
are covered by $u$ or   by $v$.

We claim that there is no edge $\{a,b\}$ in $\mathcal{E}(G)$.
Therefore, observe that the base is covered by $t$, so that
$\{a,b\}$ cannot cross $\{u,v\}$. Note that there is a triangle
crossing if $x=a$ and $b=y$ and $\{u,v\}$ crosses $\{a,b\}$ with a
triangle-crossing edge $\{u,v\}$. Edge $\{a,b\}$  crosses neither
$\{t,x\}$ nor $\{t,y\}$. If $a,b$ are distinct from $x,y$, then
there is an independent crossing of $\{t,x\}$ and $\{t,y\}$,
respectively, by $\{a,b\}$ and $\{u,v\}$. If $a=x$, then $\{t,x\}$
and $\{x,b\}$ are adjacent and do not cross and $\{x,b\}$ and
$\{u,v\}$ independently cross $\{t,y\}$ if $b \neq y$, and for
$b=y$, $\{x,y\}$ and $\{t,y\}$ cannot cross as adjacent edges.

However, after a removal of the base $\{u,v\}$, vertices $a$ and $b$
are in a common face and can be connected by an uncrossed edge $\{a,
b\}$, which clearly cannot be part of another instance of
configuration II.

Hence, we can successively remove all instances of configuration II
and every time replace the base edge by a new uncrossed edge.
\qed
\end{proof}

In consequence, we solve an open problem of Kaufmann and Ueckerdt
\cite{ku-dfang-14} on the density of fan-planar graphs and show that
  configuration II has no impact on the density.

\begin{corollary}
Adjacency-crossing and fan-crossing graphs have at most $5n-10$
edges.
\end{corollary}

\section{Conclusion}  \label{sect:conclusion}
We extended the study of fan-planar graphs initiated by Kaufmann and
Ueckerdt \cite{ku-dfang-14} and continued in \cite{bcghk-rfpg-17,
bddmpst-fan-15} and  clarified the situation around fan-crossings.
We proved that triangle-crossings can be avoided whereas
configuration II is essential for graphs but not for their density.
Thereby, we solved a problem by Kaufmann and Ueckerdt
\cite{ku-dfang-14} on the density of adjacency-crossing graphs.

Recently, progress has been made on problems for 1-planar graphs
\cite{klm-bib-17} that are still open for fan-crossing graphs, such
as (1) sparsest fan-crossing graphs, i.e., maximal  graphs with as
few edges as possible   \cite{begghr-odm1p-13} or (2) recognizing
specialized fan-crossing graphs, such as optimal fan-crossing graphs
with 5n-10 edges \cite{b-ro1plt-16}. 

In addition,
 non-simple topological graphs with multiple
edge crossings and crossings among adjacent edges have been studied
\cite{at-mneqpg-07}, and they may differ from the simple ones, as it
is known for quasi-planar graphs \cite{aapps-qpg-97}. Non-simple
fan-crossing graphs have not yet been studied.

\section{Acknowledgements}
I wish to thank Christian Bachmaier for the discussions on
fan-crossing graphs and his valuable suggestions.

\bibliographystyle{abbrv}

\begin{thebibliography}{10}

\bibitem{afps-grids-14}
E.~Ackerman, J.~Fox, J.~Pach, and A.~Suk.
\newblock On grids in topological graphs.
\newblock {\em Comput. Geom.}, 47(7):710--723, 2014.

\bibitem{at-mneqpg-07}
E.~Ackerman and G.~Tardos.
\newblock On the maximum number of edges in quasi-planar graphs.
\newblock {\em J. Comb. Theory, Ser. {A}}, 114(3):563--571, 2007.

\bibitem{aapps-qpg-97}
P.~K. Agarwal, B.~Aronov, J.~Pach, R.~Pollack, and M.~Sharir.
\newblock Quasi-planar graphs have a linear number of edges.
\newblock {\em Combinatorica}, 17(1):1--9, 1997.

\bibitem{bcghk-rfpg-17}
M.~A. Bekos, S.~Cornelsen, L.~Grilli, S.~Hong, and M.~Kaufmann.
\newblock On the recognition of fan-planar and maximal outer-fan-planar graphs.
\newblock {\em Algorithmica}, 79(2):401--427, 2017.

\bibitem{bddmpst-fan-15}
C.~Binucci, E.~{Di Giacomo}, W.~Didimo, F.~Montecchiani, M.~Patrignani,
  A.~Symvonis, and I.~G. Tollis.
\newblock Fan-planarity: Properties and complexity.
\newblock {\em Theor. Comput. Sci.}, 589:76--86, 2015.

\bibitem{b-FOL-17}
F.~J. Brandenburg.
\newblock A first order logic definition of beyond-planar graphs.
\newblock {\em J. Graph Algorithms Appl.}, 2017.
\newblock Accepted for publication.

\bibitem{b-ro1plt-16}
F.~J. Brandenburg.
\newblock Recognizing optimal 1-planar graphs in linear time.
\newblock {\em Algorithmica}, published online October 2016,
  doi:10.1007/s00453-016-0226-8.

\bibitem{begghr-odm1p-13}
F.~J. Brandenburg, D.~Eppstein, A.~Glei{\ss}ner, M.~T. Goodrich, K.~Hanauer,
  and J.~Reislhuber.
\newblock On the density of maximal 1-planar graphs.
\newblock In M.~van Kreveld and B.~Speckmann, editors, {\em {GD} 2012}, volume
  7704 of {\em {LNCS}}, pages 327--338. Springer, 2013.

\bibitem{cpkk-fan-15}
O.~Cheong, S.~Har{-}Peled, H.~Kim, and H.~Kim.
\newblock On the number of edges of fan-crossing free graphs.
\newblock {\em Algorithmica}, 73(4):673--695, 2015.

\bibitem{hm-dcgmnc-92}
H.~Harborth and I.~Mengersen.
\newblock Drawings of the complete graph with maximum number of crossings.
\newblock {\em Congressus Numerantium}, 88:225--228, 1992.

\bibitem{ku-dfang-14}
M.~Kaufmann and T.~Ueckerdt.
\newblock The density of fan-planar graphs.
\newblock {\em CoRR}, abs/1403.6184, 2014.

\bibitem{klm-bib-17}
S.~G. Kobourov, G.~Liotta, and F.~Montecchiani.
\newblock An annotated bibliography on 1-planarity.
\newblock {\em Computer Science Review}, 25:49--67, 2017.

\bibitem{ringel-65}
G.~Ringel.
\newblock Ein {S}echsfarbenproblem auf der {K}ugel.
\newblock {\em Abh. aus dem Math. Seminar der Univ. Hamburg}, 29:107--117,
  1965.

\end{thebibliography}

\end{document}